\documentclass[runningheads,envcountsame,orivec]{llncs}

\usepackage[T1]{fontenc}
\usepackage[utf8]{inputenc}
\usepackage{newtxtext}
\usepackage[varvw]{newtxmath}
\usepackage{bbding}
\usepackage{microtype}
\usepackage{algorithm}
\usepackage{algpseudocode}
\usepackage{color}
\usepackage[shortcuts]{extdash}
\usepackage{hyperref}
\usepackage[capitalise,noabbrev,nameinlink]{cleveref}

\hypersetup{
  colorlinks=true,
  linkcolor=black,
  citecolor=black,
  urlcolor=blue,
  pdfauthor={M. Utkan Gezer},
  pdftitle={Constant-Coin Complete-Information Debates for P with Arbitrarily Small Strong Error}
}

\spnewtheorem{fact}[theorem]{Fact}{\bfseries}{\itshape}

\newcommand{\Zp}{\mathbb{Z}_{+}}
\newcommand{\calP}{\mathcal{P}}
\newcommand{\AFA}{\mathsf{2AFA}}
\newcommand{\afa}{\mathsf{2afa}}
\newcommand{\dfa}{\mathsf{2dfa}}
\newcommand{\Pclass}{\mathsf{P}}
\newcommand{\pcstronglow}{\mathsf{pc\text{-}strong\text{-}low\text{-}CDEB}}
\newcommand{\concoin}{\mathsf{con\text{-}coin}}
\newcommand{\conspace}{\mathsf{con\text{-}space}}
\newcommand{\unboundedtime}{\infty\text{-}\mathsf{time}}
\newcommand{\qinit}{q_{\mathrm{init}}}
\newcommand{\qacc}{q_{\mathrm{acc}}}
\newcommand{\qrej}{q_{\mathrm{rej}}}
\newcommand{\winp}{w_{\mathrm{inp}}}
\newcommand{\wdeb}{w_{\mathrm{deb}}}
\newcommand{\pass}{\textit{pass}}
\newcommand{\Pass}{\textit{Pass}}
\newcommand{\reject}{\textit{reject}}
\newcommand{\Reject}{\textit{Reject}}
\newcommand{\accept}{\textit{accept}}
\newcommand{\Accept}{\textit{Accept}}

\floatname{algorithm}{Procedure}
\crefname{algorithm}{Procedure}{Procedures}
\Crefname{algorithm}{Procedure}{Procedures}
\crefname{fact}{Fact}{Facts}
\Crefname{fact}{Fact}{Facts}
\makeatletter
\renewcommand{\theHALG@line}{\thealgorithm.\arabic{ALG@line}}
\makeatother
\title{Constant-Coin Complete-Information Debates for \texorpdfstring{$\Pclass$}{P}
with Arbitrarily Small Strong Error}
\titlerunning{Constant-Coin Complete-Information Debates for P}
\author{M. Utkan Gezer\textsuperscript{\Envelope}\orcidID{0000-0002-5022-178X}}
\authorrunning{M. Utkan Gezer}
\institute{Department of Computer Engineering, Bo\u{g}azi\c{c}i University,\\
\.{I}stanbul, T\"urkiye\\
\email{utkan.gezer@bogazici.edu.tr}}

\begin{document}
\pagestyle{empty}
\maketitle

\begin{abstract}
We study complete\-/information debate systems in which a probabilistic finite\-/state verifier reads the alternating messages of a prover and a refuter. Demirci, Say, and Yakary{\i}lmaz showed that every language in $\Pclass$ has such debates checkable with a constant number of random bits and arbitrarily small weak error. Their strong\-/error construction, which also counts nontermination as failure, did not permit arbitrary error reduction. We close this gap: for every $L\in\Pclass$ and every $\varepsilon>0$, there is a constant\-/space verifier using a constant number of private coin tosses that has perfect completeness and strong error at most $\varepsilon$. The verifier simulates a polynomial\-/time alternating multihead finite automaton, privately spot\-/checking one of its input heads. The key observation is that, on a nonmember, the refuter may concede any round in which the prover first misreports a head reading. This ensures termination against every prover when the refuter follows the specified strategy, and permits strong\-/error reduction by repetition.
\end{abstract}

\section{Introduction}

An interactive proof system allows a computationally limited probabilistic verifier to check language membership by exchanging messages with a computationally unbounded prover that tries to convince it to accept~\cite{GMR89}. A debate system gives the verifier a second participant, the refuter, whose objective is to convince it to reject. The verifier judges the competing claims from the messages supplied by the prover and refuter, as in probabilistically checkable debate systems~\cite{CFLS95}. Among debate models with differing restrictions on the participants' access to each other's messages~\cite{R79}, we consider the complete\-/information case, in which both participants see the entire preceding debate before supplying their next message. We use the formal model of Demirci, Say, and Yakary{\i}lmaz~\cite{DSY15}, with a finite\-/state verifier. On a fixed input, the players' strategies determine the entire debate string. The verifier reads this string using its private coins and sends no messages to either player.

Even with only finite\-/state memory, a verifier can gain computational power through interaction, as shown by Dwork and Stockmeyer~\cite{DS92}. Say and Yakary{\i}lmaz imposed the additional restriction of using only constantly many random bits and introduced a technique for checking multihead automaton computations with constant space and randomness, obtaining both arbitrarily small weak error and bounded strong error for every language in $\mathsf{NL}$, the class of languages recognizable in nondeterministic logarithmic space~\cite{SY14}. Demirci, Say, and Yakary{\i}lmaz adapted this technique to complete\-/information debates by simulating alternating multihead automata, with the prover and refuter supplying existential and universal choices, respectively. They obtained the same two types of error guarantees for every language in $\Pclass$, again with constant space and randomness~\cite{DSY15}.

The distinction between weak and strong error concerns the treatment of nontermination. Weak error does not penalize a verifier that is induced to run forever on a nonmember, whereas strong error counts this as failure as well as erroneous acceptance. Neither of the above strong\-/error constructions allowed arbitrary error reduction. Gezer and Say later refined the single\-/prover construction to obtain arbitrarily small strong error for a subclass of $\mathsf{NL}$~\cite{GS22}. Like the earlier weak\-/error constructions, this refinement has perfect completeness. This leaves the analogous question for complete\-/information debates: to what extent can strong error be reduced arbitrarily while retaining constant space and randomness?

We obtain this improvement for every language in $\Pclass$. Our verifier has perfect completeness: it accepts every member with certainty against every refuter when paired with an appropriate prover. It also has arbitrarily small strong error: for any prescribed $\varepsilon>0$, the verifier can be constructed so that, on a nonmember, every prover faces a refuter that makes it reject with probability at least $1-\varepsilon$. In each of these cases, the specified player's strategy guarantees that the verifier halts in polynomial time regardless of the opponent's behavior. This result and construction first appeared in the author's doctoral thesis~\cite[Section~6]{G26}.

The protocol builds on the earlier spot\-/checking construction by allowing the refuter to concede a round rather than continue a simulation compromised by a false report. Against the refuter constructed in the proof, these concessions prevent undetected false reports from trapping the verifier in an infinite round, making repetition effective for reducing strong error. \Cref{sec:main} gives the full protocol and analysis.

The remaining sections are as follows: \Cref{sec:preliminaries} introduces alternating multihead automata, complete\-/information debates, and constant\-/randomness sampling. \Cref{sec:pruning} establishes polynomial\-/time pruning, and \Cref{sec:conclusion} discusses possible extensions.

\section{Preliminaries}\label{sec:preliminaries}

Let $\Zp$ denote the set of positive integers. For integers $a\le b$, $[a..b]$ denotes the set of all integers from $a$ to $b$, inclusive. For a set $S$, $\calP(S)$ denotes its power set. An alphabet $\Sigma$ is a finite nonempty set of symbols, and $\Sigma^*$ is the set of finite strings over $\Sigma$, including the empty string. We write $\winp\in\Sigma^*$ for the input string and $n=|\winp|$ for its length. Let $\triangleright$ and $\triangleleft$ be distinct symbols outside $\Sigma$, serving as the left and right endmarkers, respectively. The endmarker\-/extended input alphabet is $\Sigma_{\bowtie}=\Sigma\cup\{\triangleright,\triangleleft\}$. The set $D=\{-1,0,1\}$ specifies a head movement one cell left, no movement, or one cell right, respectively.

\subsection{Alternating Multihead Finite Automata}

An alternating two\-/way finite automaton with $k$ input heads, abbreviated $\afa(k)$, is a tuple
\[
 M=(Q_{\exists},Q_{\forall},\Sigma,\delta,\qinit,\qacc,\qrej),
\]
where $Q_{\exists}$ and $Q_{\forall}$ are disjoint finite sets of existential and universal states, respectively, and $\qacc$ and $\qrej$ are distinct accepting and rejecting halting states outside those sets. The full state set is $Q=Q_{\exists}\cup Q_{\forall}\cup\{\qacc,\qrej\}$, and the transition function is
\[
 \delta:(Q_{\exists}\cup Q_{\forall})\times\Sigma_{\bowtie}^{k}
 \longrightarrow \calP\bigl(Q\times D^{k}\bigr).
\]
The machine begins in its designated initial state $\qinit\in Q_{\exists}\cup Q_{\forall}$. Its read\-/only input tape contains $\triangleright\winp\triangleleft$, and all $k$ heads begin on $\triangleright$. When the state is $q$ and the symbols scanned by heads $H_1,\ldots,H_k$ are $\sigma_1,\ldots,\sigma_k$, the set $\delta(q,\sigma_1,\ldots,\sigma_k)$ lists the available actions. Each action $(q',d_1,\ldots,d_k)$ is a directive to change the state to $q'$ and move each head $H_i$ by $d_i$. One of these actions is selected and executed in a step. Upon entering $\qacc$ or $\qrej$, the machine halts and accepts or rejects, respectively.

Without loss of language recognition power, we require that no action move a head beyond an endmarker and that $\delta(q,\sigma_1,\ldots,\sigma_k)\neq\varnothing$ for every nonhalting state $q$ and every tuple of scanned symbols.

A configuration on a fixed input is a tuple $(q,p_1,\ldots,p_k)$ consisting of the current state and the head positions, indexed from $0$ at $\triangleright$ through $n+1$ at $\triangleleft$. There are exactly $|Q|\cdot(n+2)^k$ configurations, some of which may be unreachable from the initial configuration. The execution tree of $M$ on $\winp$ has the initial configuration $(\qinit,0,\ldots,0)$ at its root. At each nonhalting node, it has one child for each available action, labeled by the resulting configuration. Halting nodes have no children.

An \emph{existential strategy} is a subtree of the execution tree rooted at the initial configuration, retaining exactly one child at each existential node and all children at each universal node. Dually, a \emph{universal strategy} retains all children at each existential node and exactly one child at each universal node. An existential (respectively, universal) strategy is called \emph{winning} if it is finite and all its leaves are accepting (respectively, rejecting) configurations. The \emph{height} of a finite strategy is the maximum number of transitions on a root\-/to\-/leaf path.

The machine is said to \emph{accept} an input if it has a winning existential strategy, to \emph{reject} if every existential strategy contains a rejecting configuration, and to \emph{loop} otherwise. It \emph{recognizes} $L$ if it accepts exactly the strings in $L$; it may reject or loop on nonmembers.

A deterministic $k$\-/head automaton, denoted $\dfa(k)$, is the special case with exactly one available action at each nonhalting configuration. Its execution tree is a single path, so existential and universal states make no difference to its outcome.

Let $\AFA(k)$ denote the class of languages recognized by $\afa(k)$'s. The notation $\AFA(k,O(n^t)\text{-time})$ restricts this class to machines whose every computation path, on every input of length $n$, halts within $O(n^t)$ steps.

We use two results of King~\cite{K88}.

\begin{fact}[{\cite[Theorem~3.1]{K88}}]\label{fact:king-p}
$\bigcup_{k\in\Zp}\AFA(k)=\Pclass$.
\end{fact}

\begin{fact}[{\cite[Lemma~3.9]{K88}}]\label{fact:king-height}
Let $M$ be a $\afa(k)$ with state set $Q$. On every member input of length $n$, $M$ has a winning existential strategy of height at most $|Q|\cdot(n+2)^k$.
\end{fact}

The following is a standard consequence of determinacy for finite reachability games, applied to the configuration graph with rejecting configurations as the target set~\cite[Section~4.2, pp.~103--104]{G11}.

\begin{fact}\label{fact:universal-strategy}
If an alternating multihead finite automaton $M$ rejects an input $\winp$, then it has a winning universal strategy on $\winp$.
\end{fact}

\subsection{Complete-Information Debate Systems}

A debate system has three agents: a probabilistic verifier $V$, a prover $P$, and a refuter $R$. The prover and refuter, collectively called the \emph{players}, argue for membership and nonmembership, respectively, of the common input $\winp$ in a language $L$. Our verifiers have finite\-/state memory, a single two\-/way head on the read\-/only input tape containing $\triangleright\winp\triangleleft$, and a single one\-/way head on a read\-/only debate tape. A two\-/way head may move one cell left or right or stay in place; a one\-/way head may move one cell right or stay in place. The input and debate heads initially scan $\triangleright$ and the first debate symbol, respectively. The verifier uses the outcomes of independent fair coin tosses to make its random choices.

Let $\Gamma_1$ and $\Gamma_0$ be disjoint finite alphabets for prover and refuter messages, respectively, and let $\Gamma=\Gamma_1\cup\Gamma_0$. Both players are computationally unbounded. In a \emph{complete\-/information debate system} (CDS), their strategies are functions $P:\Sigma^*\times\Gamma^*\to\Gamma_1$ and $R:\Sigma^*\times\Gamma^*\to\Gamma_0$, used on prefixes of even and odd length, respectively. On input $\winp$, these functions determine an infinite debate string $\wdeb=\gamma_1\gamma_2\cdots$: for odd $j$, $\gamma_j=P(\winp,\gamma_1\cdots\gamma_{j-1})$, and for even $j$, $\gamma_j=R(\winp,\gamma_1\cdots\gamma_{j-1})$. Thus the players alternate, beginning with the prover, and each sees the entire debate prefix before supplying its next symbol. The debate tape initially contains the resulting string $\wdeb$.

For an input $\winp$ and players $P,R$, let
\[
 \operatorname{acc}_{V,P,R}(\winp),\quad
 \operatorname{rej}_{V,P,R}(\winp),\quad\text{and}\quad
 \operatorname{loop}_{V,P,R}(\winp)
\]
be the corresponding acceptance, rejection, and nontermination probabilities, taken over the verifier's private coins. Their sum is $1$. We say that $V$ recognizes $L$ with perfect completeness and strong error at most $\varepsilon$ if
\begin{align*}
&\forall \winp\in L\;\exists P\;\forall R:
  \operatorname{acc}_{V,P,R}(\winp)=1,\\
&\forall \winp\notin L\;\forall P\;\exists R:
  \operatorname{rej}_{V,P,R}(\winp)\geq 1-\varepsilon.
\end{align*}
The second condition bounds the sum of false acceptance and nontermination probabilities by $\varepsilon$. Under the weak\-/error criterion, it is replaced by the requirement that $\operatorname{acc}_{V,P,R}(\winp)\leq\varepsilon$, with the same quantifiers; nontermination on a nonmember is then not counted as error.

We denote by
\[
 \pcstronglow(\conspace,\unboundedtime,\concoin)
\]
the class of languages $L$ for which, for every $\varepsilon>0$, there is a CDS verifier recognizing $L$ with perfect completeness and strong error at most $\varepsilon$, using at most a constant amount of space and tossing at most a constant number of coins. These space and coin\-/toss upper bounds may depend on $L$ and $\varepsilon$, but hold for every input, every pair of player strategies, and every sequence of coin outcomes. Constant space is realized by finite\-/state memory. The prefix $\mathsf{pc}$ denotes perfect completeness, and $\mathsf{strong\text{-}low}$ denotes arbitrarily small strong error. The qualifier $\unboundedtime$ imposes no time bound; nontermination is permitted but remains subject to the strong\-/error condition above.

\subsection{Picking One of Constantly Many Integers}

The verifier will use the following elementary sampler. Its explicit probability bounds are useful because the number of choices need not be a power of two.

\begin{lemma}\label{lem:pick-int}
For every $K\in\Zp$, a coin\-/tossing finite automaton can pick $i\in[1..K]$ in constant time depending on $K$, using $r=\lceil\log_2 K\rceil$ random bits, so that every $i$ is selected with probability $p_i$ satisfying
\[
 \frac{1}{2K}<p_i<\frac{2}{K}.
\]
\end{lemma}

\begin{proof}
Toss $r$ coins to obtain a uniformly random integer $v\in[0..2^r-1]$ and select $i=(v\bmod K)+1$. If $K$ is a power of two, each index has one preimage and is selected with probability $1/K$. Otherwise, each index has one or two preimages, giving probabilities $1/2^r$ or $2/2^r$, respectively, which satisfy
\[
 \frac{1}{2K}<\frac{1}{2^r}<\frac{1}{K}<\frac{2}{2^r}<\frac{2}{K}.
\]
For fixed $K$, the sampling and the map from $v$ to $i$ can be implemented in finite control in constant time. For $K=1$, select the sole index without tossing any coins.
\qed
\end{proof}

\section{Polynomial-Time Alternating Multihead Automata}\label{sec:pruning}

\Cref{fact:king-height} bounds the height of a winning strategy on members, but an arbitrary $\afa(k)$ may still have infinite computation paths. We first implement a polynomial timer with input heads and then use it to truncate branches that outlast it.

\begin{lemma}[Polynomial timer]\label{lem:timer}
For all $c_1,c_2,k\in\Zp$, there is a $\dfa(k)$ that runs for at least $c_1\cdot(n+c_2)^k$ steps and halts in $O(n^k)$ time on inputs of length $n$.
\end{lemma}

\begin{proof}
Use a $k$\-/digit counter. For each $i\in[1..k]$, represent digit $i$ by $H_i$'s position and a spill counter $s_i\in[0..c_2]$, with $s_i=0$ initially. To increment digit $i$, move $H_i$ one cell right if it is not on $\triangleleft$; otherwise increment $s_i$. When $s_i$ reaches $c_2$, the timer terminates if $i=k$; otherwise return $H_i$ to $\triangleright$, reset $s_i$ to zero, and carry by incrementing digit $i+1$. Repeatedly pause for $c_1$ steps and increment the counter by incrementing digit $1$ and propagating carries. All spill counters and control information fit in finite control.

Each digit requires $b=n+1+c_2$ increments before a carry or termination. There are therefore $b^k$ counter increments, each preceded by a pause, giving at least $c_1 b^k\geq c_1\cdot(n+c_2)^k$ steps.

Excluding head resets, each counter increment and its preceding pause take constant time, for $O(b^k)=O(n^k)$ time in total. For $i<k$, head $H_i$ is reset $b^{k-i}\in O(n^{k-i})$ times, at $O(n+1)$ steps each, totaling $O(n^{k-i+1})$ time. The total running time is therefore
\[
 O(n^k)+\sum_{i=1}^{k-1}O(n^{k-i+1})=O(n^k).\tag*{\(\qed\)}
\]
\end{proof}

\begin{lemma}[Polynomial-time pruning]\label{lem:pruning}
For every $k\in\Zp$,
\[
 \AFA(k)\subseteq\AFA(2k,O(n^k)\text{-time}).
\]
\end{lemma}

\begin{proof}
We apply the pruning method used in~\cite[Lemma~3]{GS22} to alternating multihead automata. Let $L\in\AFA(k)$ and let $M$ be a $\afa(k)$ recognizing $L$, with state set $Q$. Apply \Cref{lem:timer} with $c_1=|Q|$, $c_2=2$, and the same $k$ to obtain a deterministic $k$\-/head timer $N$. Construct a $2k$\-/head automaton $A$ running \Cref{proc:pruned-tree}. It uses its first $k$ heads to simulate $M$ and its last $k$ heads to simulate $N$. Its finite control stores the pair of simulated states; each nonhalting product state is existential or universal according to its $M$ component.

\begin{algorithm}[htbp]
\caption{Polynomial-time pruning}\label{proc:pruned-tree}
\raggedright
\begin{algorithmic}[1]
\State Initialize $M$ and $N$.
\State Until either simulation halts, simultaneously execute an available action of $M$ on the first $k$ heads and the unique action of $N$ on the last $k$ heads, using the branching type of $M$.
\State Accept if $M$ accepted, including when $N$ halts on the same transition; reject otherwise.
\end{algorithmic}
\end{algorithm}

By \Cref{lem:timer}, $N$ runs for at least
\[
 f(n)=|Q|\cdot(n+2)^k
\]
steps and for at most $O(n^k)$ steps.

For $\winp\in L$, \Cref{fact:king-height} supplies a winning existential strategy for $M$ of height at most $f(n)$. Following the same existential choices in $A$ makes the simulation of $M$ accept no later than the timer expires, regardless of the universal choices. Thus $A$ accepts $\winp$.

For $\winp\notin L$, $M$ does not accept. Since $A$ preserves $M$'s branching and only truncates paths by rejecting, $A$ does not accept either. The timer ensures that every path halts, so $A$ rejects $\winp$.

On every input of length $n$, each simulation transition of $A$ advances $N$ by one step, and the simulation stops no later than $N$ halts. Initialization and the final decision take constant time, so every computation path of $A$ halts within $O(n^k)$ steps. Hence $A$ recognizes $L$ with $2k$ heads in the required time.
\qed
\end{proof}

\section{Constant-Coin Debates for \texorpdfstring{$\Pclass$}{P}}\label{sec:main}

We use the verifier construction of Demirci, Say, and Yakary{\i}lmaz~\cite[Lemma~6]{DSY15}, with a modification that allows repetition to reduce strong error. To verify membership in a language $L\in\Pclass$, their verifier simulates an alternating multihead automaton recognizing $L$. The prover reports the head readings and chooses existential actions, while the refuter chooses universal actions. The verifier privately selects one head to track, checks its reported readings, and repeats the simulation in a fixed number of rounds.

Our modification allows the refuter to concede a round at its discretion. A concession following a false head\-/reading report ends the round if the verifier misses the lie; otherwise, the verifier rejects. This prevents undetected false reports from blocking further rounds and enables strong\-/error reduction by repetition.

\begin{theorem}\label{thm:main}
\[
 \Pclass\subseteq\pcstronglow(\conspace,\unboundedtime,\concoin).
\]
\end{theorem}

\begin{proof}
Let $L\in\Pclass$ and $\varepsilon\in(0,1)$ be arbitrary. By \Cref{fact:king-p} and \Cref{lem:pruning}, choose a polynomial\-/time alternating multihead automaton
\[
 M:=M_L=(Q_{\exists},Q_{\forall},\Sigma,\delta,\qinit,\qacc,\qrej)
\]
recognizing $L$, independently of $\varepsilon$. Let $k:=k_M$ be half the number of heads of $M$. By the pruning construction, $k$ is a positive integer and $M$ halts within $O(n^k)$ steps on every computation path. We construct a verifier $V:=V_{M,\varepsilon}$ that checks repeated simulations of $M$.

\paragraph{Debate Format and Verifier.}
Each simulated transition, whether existential or universal, is described by a prover--refuter symbol pair on the debate tape. The prover symbol contains all $2k$ head readings and a proposed action; the refuter symbol contains a proposed action or a concession. Unless the refuter concedes, the verifier uses the prover's action at an existential state and the refuter's action at a universal state, ignoring the other action. If the refuter concedes and the verifier detects no false head reading, the round ends without executing either action.

Let
\[
 \mathcal{A}=\bigl(Q_{\exists}\cup Q_{\forall}\cup\{\qacc,\qrej\}\bigr)\times D^{2k}
\]
be the set of possible action tuples. The message alphabets are
\[
 \Gamma_1=\Sigma_{\bowtie}^{2k}\times\mathcal{A}
 \qquad\text{and}\qquad
 \Gamma_0=\mathcal{A},
\]
for the prover and refuter, respectively. The refuter can concede a round using
\[
 a_{\mathrm{pass}}=(\qacc,0,\ldots,0)
\]
regardless of the simulated state or whether this action is available.

Set
\[
 \alpha:=\alpha_k=1-\frac{1}{4k},
 \qquad
 m:=m_{k,\varepsilon}=\left\lceil\log_{\alpha}\varepsilon\right\rceil,
 \qquad
 r:=r_k=\left\lceil\log_2(2k)\right\rceil.
\]
\Cref{proc:spot-check} describes the verifier's computation, which consists of at most $m$ rounds. Each round starts a new simulation of $M$ from its initial configuration and uses $r$ coin tosses to select a head to track. The command \pass{} ends the current round immediately, skipping its remaining instructions; the verifier starts the next round, or accepts if all $m$ rounds are complete. It is invoked when the simulated machine reaches $\qacc$ or the refuter concedes; an invalid refuter action at a universal state is treated as a concession. The commands \accept{} and \reject{} halt the verifier with the corresponding decision on the input.

\begin{algorithm}[htbp]
\caption{Constant-coin debate verification}\label{proc:spot-check}
\raggedright
\begin{algorithmic}[1]
\For{$m$ rounds}
  \State Return the verifier's input head to $\triangleright$ without consuming debate symbols, and set the simulated state to $q\leftarrow\qinit$.
  \State Use the sampler from \Cref{lem:pick-int} to select $i\in[1..2k]$ privately at random.
  \While{$q\notin\{\qacc,\qrej\}$}
    \State Consume the next prover symbol $(\sigma_1,\ldots,\sigma_{2k},A_1)\in\Gamma_1$ and refuter symbol $A_0\in\Gamma_0$ from the debate tape, and read the input symbol $\sigma$ under the verifier's head.
    \State \Reject{} if $\sigma\neq\sigma_i$. \Pass{} if $A_0=a_{\mathrm{pass}}$.
    \State If $q\in Q_{\exists}$: \Reject{} if $A_1\notin\delta(q,\sigma_1,\ldots,\sigma_{2k})$; otherwise, set $(q',d_1,\ldots,d_{2k})\leftarrow A_1$.
    \State If $q\in Q_{\forall}$: \Pass{} if $A_0\notin\delta(q,\sigma_1,\ldots,\sigma_{2k})$; otherwise, set $(q',d_1,\ldots,d_{2k})\leftarrow A_0$.
    \State Set $q\leftarrow q'$ and move the verifier's input head by $d_i$.
  \EndWhile
  \State If $q=\qacc$, \pass{}; otherwise, \reject{}.
\EndFor
\State \Accept{}.
\end{algorithmic}
\end{algorithm}

The verifier stores the simulated state and physically tracks only head $i$. Since it never executes an invalid action, the simulation follows a genuine computation of $M$ until an undetected false head report occurs.

As long as the verifier continues, its simulated state and round boundaries are determined by the debate prefix, independently of $i$. The players can therefore follow the simulation without knowing the selected head. Before the first false report in each round, the input and debate prefix also determine the genuine configuration of $M$ and its path in the execution tree.

\paragraph{Perfect Completeness.}
For $\winp\in L$, choose a winning existential strategy $T_{\exists}$ of $M$ on $\winp$. Define a prover $P$ that starts at the root of $T_{\exists}$ in each round. At a node with configuration $c$, it supplies $(\sigma_1,\ldots,\sigma_{2k},A_1)$, where the $\sigma_i$ are the actual head readings in $c$. At an existential node, $A_1$ is the action leading to the unique child retained in $T_{\exists}$; at a universal node, $P$ supplies the ignored tuple $(\qacc,0,\ldots,0)$. After each simulated transition, $P$ moves to the child specified by the action used in that transition.

Every valid refuter action at a universal node leads to a child in $T_{\exists}$. A concession or invalid universal action passes the round, and $P$ returns to the root if another round begins. Thus the rule defines $P$ throughout all rounds against any refuter. The truthful reports and valid existential actions rule out rejection during the simulation. Any round not passed early reaches an accepting leaf because $T_{\exists}$ is finite. Thus all $m$ rounds pass, and $V$ accepts with probability $1$.

\paragraph{Strong Soundness.}
For $\winp\notin L$, $M$ rejects because it recognizes $L$ and halts on every path. Choose a winning universal strategy $T_{\forall}$ on $\winp$ by \Cref{fact:universal-strategy}. Define a refuter $R$ that starts at its root in each round and updates its current node after each simulated transition. On reading a prover symbol, $R$ compares the reported readings with the actual readings in that node's configuration. If any reading is false, it sends $a_{\mathrm{pass}}$ to concede the round, making $V$ pass if it fails to detect the false report. Otherwise, at a universal node, it sends the action leading to the unique child retained in $T_{\forall}$; at an existential node, it sends the ignored tuple $(\qrej,0,\ldots,0)$. Note that this construction of $R$ ensures that $V$'s simulation of $M$ in each round stays within $T_{\forall}$, regardless of the prover it faces.

Since $R$'s universal choices are valid and $T_{\forall}$ has only rejecting leaves, neither an invalid universal action nor reaching $\qacc$ can occur, leaving concession as the only way for $V$ to pass a round. Thus, for $V$ to accept, $R$ must concede in all $m$ rounds. Each concession follows a false report from the prover. Hence, for $V$ to accept, the prover must make a false report in all $m$ rounds.

Fix any such prover $P^*$. This fixes its debate with $R$, with rounds delimited by $R$'s concessions. Consider any of these rounds, and let $j$ be the smallest index of a head misreported in its false report. All earlier head readings in the round are truthful (an earlier false report would already have prompted $R$ to end the round by conceding), so $V$ detects no mismatch whichever head it tracks. It reads the same actions and simulates the same configurations for every choice of $i$. Hence the false report and the heads misreported in it, including $j$, do not depend on $i$. Upon reaching the round, $V$ samples $i$ using fresh private coins, independently of $j$ and of the coins used in earlier rounds. Conditional on reaching that round, \Cref{lem:pick-int} therefore gives
\[
 \Pr(i=j)>\frac{1}{4k}.
\]
If $i=j$, $V$ detects the lie and rejects. If no misreported head is selected, $V$ processes $R$'s concession and passes the round. The conditional probability that a lie escapes detection is therefore at most $1-\frac{1}{4k}=\alpha$. Applying this bound successively, the probability of passing all $m$ rounds is at most
\[
 \alpha^m
 =\alpha^{\lceil\log_{\alpha}\varepsilon\rceil}
 \leq\alpha^{\log_{\alpha}\varepsilon}
 =\varepsilon,
\]
where the inequality uses $0<\alpha<1$.

Every round terminates: without a false report, $V$ rejects an invalid existential action or reaches a rejecting leaf; a false report ends the round by rejection or concession. Thus $V$ has zero nontermination probability against $R$, and
\[
 \operatorname{rej}_{V,P^*,R}(\winp)
 \geq 1-\alpha^m
 \geq 1-\varepsilon.
\]
This proves strong soundness.

\paragraph{Resources.}
Recall that resource bounds must hold for every input, player pair, and coin outcome, including nonhalting computations possible with players other than those used above.

At most $m$ rounds begin, each using $r$ coin tosses, for a total of at most $mr$. The verifier can be implemented entirely within finite control, without a work tape: the round counter, selected head's label $i\in[1..2k]$, simulated state, and any temporarily retained readings or action tuples all range over finite sets depending only on $L$ and $\varepsilon$. Its physical input head tracks the selected head's input position. These constant space and coin bounds, together with perfect completeness and strong soundness, prove the theorem.
\qed
\end{proof}

\begin{remark}
In~\cite[Lemma~6]{DSY15}, an undetected false report can prevent further rounds by making the simulation run forever. Concessions eliminate this possibility against the refuter constructed above, enabling repetition to reduce strong error.
\end{remark}

\begin{remark}[Running time]\label{rem:running-time}
On $\winp\in L$, the constructed prover guarantees termination against every refuter; on $\winp\notin L$, the constructed refuter guarantees termination against every prover. In either case, each round follows a genuine computation of $M$ until it ends or encounters a false report, which ends the round. Processing each simulated transition takes constant time, so a round uses $O(n^k)$ steps, plus $O(n+1)$ to reset the input head. Since $m$ is constant, the total time is $O(n^k)$ when the constructed prover is used on members or the constructed refuter is used on nonmembers.

For arbitrary player pairs, an undetected false report without a concession can corrupt the simulated timer. Thus $V$ has no finite running\-/time bound over all player pairs and may run forever.
\end{remark}

\section{Conclusions and Further Questions}\label{sec:conclusion}

We have shown that every language in $\Pclass$ has complete\-/information debates checkable by a finite\-/state verifier with perfect completeness, arbitrarily small strong error, and only constantly many private coin tosses. The improvement over the weak\-/error analysis of Demirci, Say, and Yakary{\i}lmaz~\cite[Lemma~6]{DSY15} comes from allowing the refuter to forgo a compromised round, thereby preventing a dishonest prover from converting an undetected spot\-/check failure into nontermination against the prescribed refuter.

The same idea may be useful in other debate settings where one truthful participant can recognize that the opposing participant has supplied inconsistent information. A direct extension to the partial- and zero\-/information protocols of Demirci, Say, and Yakary{\i}lmaz~\cite[Lemma~10 and Theorem~19]{DSY15}, however, faces an information barrier. Universal choices hidden from the prover can change the simulated heads' positions. The refuter alone has the information needed to report or verify the corresponding readings, yet cannot be relied upon to expose its own false reports. This prevents a direct use of concessions triggered by detecting false readings, rather than ruling out concessions altogether. A direction for future research is to determine whether concessions based on other criteria can permit safe round termination under restricted visibility while retaining constant space and randomness.

In the complete\-/information protocol, the refuter can detect a false head reading but may choose to continue the round. If the verifier misses the false reading, the round may run forever (\Cref{rem:running-time}). This raises the question of whether \Cref{thm:main} still holds if the prescribed strong\-/error bound $\varepsilon$ must also bound the probability of nontermination on every input for every pair of player strategies. This would still allow some runs to continue forever. Can we guarantee that every run halts, with a finite bound on the number of steps for each input length?

\begin{credits}
\subsubsection{\ackname}

I am grateful to A. C. Cem Say, my master's and doctoral thesis advisor, for his guidance, encouragement, and insightful discussions throughout this work. This research began with his suggestion to investigate whether the error\-/reduction ideas in my master's thesis could be applied to the debate systems of Demirci, Say, and Yakary{\i}lmaz~\cite{DSY15}. I also thank the anonymous reviewers for their comments.

\subsubsection{\discintname}
The author has no competing interests to declare that are relevant to the content of this article.

\end{credits}

\bibliographystyle{splncs04}
\bibliography{paper/references}

\end{document}